%% file: main.tex
\documentclass{article}

\usepackage[english]{babel}
\usepackage{amsmath}
\usepackage{amsfonts}
\usepackage{amsthm}
\usepackage{caption}
\usepackage{bbm}
\usepackage[utf8]{inputenc}
\usepackage[T1]{fontenc}
\usepackage{inconsolata}
\usepackage{authblk}

\usepackage{tikz}
\usepackage{tikz-cd}
\usetikzlibrary{shapes,arrows}
\usetikzlibrary{positioning}

\newcommand{\bfW}{\mathbf{W}}
\newcommand{\bfw}{\mathbf{w}}

\newcommand{\RR}{\mathbb{R}}
\newcommand{\Ss}{\mathcal{S}}

\newcommand{\T}{\mathcal{T}}
\newcommand{\VV}{\mathcal{V}}

\newcommand{\NN}{\mathcal{N}}
\newcommand{\e}{\mathbf{e}}
\newcommand{\x}{\mathbf{x}}

\newcommand{\h}{\mathbf{h}}
\newcommand{\uu}{\mathbf{u}}
\newcommand{\vv}{\mathbf{v}}
\newcommand{\F}{\mathcal{F}}

\newcommand{\dP}{\mathcal{P}}
\newcommand{\D}{\partial}
\newcommand{\DD}{\mathcal{D}}
\newcommand{\sumd}{\tau}

\newtheorem{definition}{Definition}[section]

\newtheorem{theorem}{Theorem}[section]
\newtheorem{remark}{Remark}[section]
\newtheorem{corollary}{Corollary}[theorem]
\newtheorem{proposition}{Proposition}[section]

\usepackage[sorting=none,backend=biber]{biblatex}

\title{Operational Calculus for Differentiable Programming}
\date{}

\author[1]{Žiga Sajovic}
\author[2]{Martin Vuk }

\affil[1]{\footnotesize XLAB d.o.o.\authorcr ziga.sajovic@xlab.si}
\affil[2]{\footnotesize University of Ljubljana, Faculty of Computer and Information Science \authorcr martin.vuk@fri.uni-lj.si}

\begin{document}

\maketitle

\input{./latex/abstract.tex}

\input{./latex/Introduction.tex}
\input{./latex/Computer_Programs_as_Maps_on_a_Vector_Space.tex}

\input{./latex/Differentiable_Maps_and_Programs.tex}

\input{./latex/Differentiable_Programming_Spaces.tex}

\input{./latex/Operational_Calculus_on_Programming_Spaces.tex}

\input{./latex/Iterators.tex}

\input{./latex/Conclusions.tex}

\printbibliography

\end{document}

%% file: latex/abstract.tex
\begin{abstract}
  In this work we present a theoretical model for differentiable programming. We
  construct an algebraic language that encapsulates formal semantics of
  differentiable programs by way of \emph{Operational Calculus}. The algebraic
  nature of Operational Calculus can alter the properties of the
  programs that are expressed within the language and transform them into their
  solutions.

  In our model, programs are elements of \emph{programming spaces} and viewed as
  maps from the \emph{virtual memory space} to itself. Virtual memory space is
  also an algebra of programs, \emph{an algebraic data structure} one can
  calculate with. We define the \emph{operator of differentiation} ($\D$) on
  programming spaces and, using its powers, implement the \emph{general shift
    operator}. We provide the formula for the expansion of a differentiable
  program into an infinite tensor series in terms of the powers of $\D$ and
  implement a differentiable composition of differentiable programs by
  expressing the \emph{operator of program composition} in terms of the
  generalized shift operator and $\D$. The presented operators serve as
  an abstraction and act as the main components of our language.
  
  We demonstrate our model's usefulness in differentiable programming by using it
  to analyse iterators, deriving \emph{fractional iterations} and their
  \emph{iterating velocities}, and explicitly solve the special
  case of \emph{ReduceSum}.
 \end{abstract}

%% file: latex/Introduction.tex
\section{Introduction}

According to John Backus, Von Neumann languages do not have useful properties for reasoning about programs. Axiomatic and denotational semantics are precise tools for describing and understanding conventional programs, but they only talk about them and cannot alter their ungainly
properties \cite{backus}. This issue has partially been addressed by algebraic data types employed by functional programming, where a mapping has been shown between grammars and semirings \cite{7Trees}.
Yet due to the lack of inverses (hence the semiring structure) they remain limited in the algebraic manipulations they are allowed to employ \cite{complexCat}.

As computer programs are the dominant tool for modern problem solving, the need
for examining the analytic properties of programs led to the development of
various tools for dealing with derivatives of computer programs (finite
difference methods, automatic differentiation). Yet the developed techniques are only efficient ways of calculating derivatives, and do not construct any meaningful algebraic structure over differentiable programs. As such, there is still a need for a framework that properly captures the analytic properties of differentiable programs and provides higher-order constructs that can reason about them. Such a framework can be provided by \emph{Operational Calculus}, because \emph{unlike von Neumann languages, the language of ordinary algebra is suitable both for stating its laws and for transforming an equation into its solution, all within the language} \cite{backus}.

The ideas of functional programming and automatic differentiation have been
combined to some extent successfully in the field of Deep Learning for example.
It has shown itself to be more than a collection of machine learning algorithms
and the name \emph{Differentiable Programming} emerged as a new programming
paradigm. But because the field is still in its youth, most of the advances
come as a result of empirical investigations. Yet, as it is founded on rigorous
mathematical objects, it offers an opportunity to be formalized as an algebraic
language. Mathematical analysis and calculus found their way into programming,
where different fields employ analytic properties of programs. What seems to be
lacking in these attempts is a mechanism that would facilitate revealing
transformations of these properties, while abstracting away the gory details of
calculus. 

The proposed theoretical model and the constructed operational calculus aim to
fill this gap. We have been inspired by the development of differentiable programming to
formalize a theoretical model, that encompasses the ideas underlying
differentiable programming and provides a general setting for the study of
differentiable programs. The presented theoretical model enables analytic
investigations of differentiable programs through algebraic tools, that are
closer to the field of programming; i.e. the presented operators can take the
same role as higher order functions in functional programming. We introduce a
\emph{Virtual Tensor Machine} as a language that extends functional definition
of programs with a \emph{Tensor Series Algebra} of the memory. Such a tensor
description of the memory can also serve as a formalization of recent
advancements in high performance computing hardware, ex. tensor processing
units by Google and tensor cores by Nvidia. This algebraic structure inherent
to our model allows us to establish an \emph{Operational Calculus} of
higher-order constructs that can facilitate reasoning about differentiable
programs. Furthermore, the presented model is self-sufficient, as the
Operational Calculus presented herein is implemented strictly within the
language itself. We hope, that the introduction of Operational Calculus into the
field of computer science will provide a new approach to solving problems and offer
a different view on the field, as it has already done in modern physics \cite{OpCalc}.

We demonstrate our frameworks usefulness to differentiable programming by using it to
analyse iterators of differentiable programs and derive their \emph{fractional 
iterations}. This allowed us to derive their \emph{iterating velocities}, ie. higher 
order derivatives of the iterates with respect to the number of iterations, which may 
offer new insights into iterated processes that feature prominently in programming. We 
use these ideas to explicitly solve the special case of \emph{ReduceSum} and its (higher-
order) iterating velocities.

%% file: latex/Computer_Programs_as_Maps_on_a_Vector_Space.tex
\section{Computer Programs as Maps on a Vector Space}

We will model computer programs as maps on a vector space. If
we only focus on the real valued variables (of type \texttt{float} or
\texttt{double}), the state of the memory can be seen as a high
dimensional vector\footnote{We assume the variables of interest to be of type \texttt{float} for
  simplicity. Theoretically any field can be used instead of $\RR$.}. 
A set of all the possible states of the program's memory,
can be modeled by a finite dimensional real vector space $\VV\equiv \RR^n$. We
will call $\VV$ the \emph{memory space of the program}. The effect of a computer
program on its memory space $\VV$, can be described by a map
\begin{equation}
  \label{eq:map}
  P:\VV\to \VV.
\end{equation}
A programming space is a space of maps $\VV\to\VV$ that can be implemented as a
program in a specific programming language. 
\begin{definition}[Euclidean machine] The tuple $(\VV,\F)$ is an Euclidean machine, where
  \begin{itemize}
  \item
  $\VV$ is a finite dimensional vector space over a complete field $K$, serving
  as memory\footnote{In most applications the field $K$ will
    be $\RR$}
  \item
  $\F< \VV^\VV$ is a subspace of the space of maps $\VV\to \VV$, called the \emph{programming space}, serving as actions on the memory.
  \end{itemize}  
\end{definition}

At first glance, the \emph{Euclidean machine} seems like a description of functional programming, with its compositions inherited from $\F$. An intended impression, as we wish for the \emph{Euclidean machine} to build on its elegance. But note that in the coming section an additional restriction is imposed on $\F$; that of its elements being differentiable.

%% file: latex/Differentiable_Maps_and_Programs.tex
\section{Differentiable Maps and Programs}

To define differentiable programs, let us first recall some
definitions from multivariate calculus.
\begin{definition}[Derivative]
  Let $V,U$ be Banach spaces. A map $P:V\to U$ is differentiable at a point
  $\x\in V$, if there exists a linear bounded operator $TP_\x:V\to U$ such that
  \begin{equation}
    \label{eq:frechet}
    \lim_{\h\to 0}\frac{\|P(\x+\h)-P(\x)-TP_\x(\h)\|}{\|\h\|} = 0.
  \end{equation}
  The map $TP_\x$ is called the \emph{Fréchet derivative} of the map $P$ at the
  point $\x$.
\end{definition}
For maps $\RR^n\to \RR^m$ Fréchet derivative can be expressed by multiplication
of vector $\h$ by the Jacobi matrix $\mathbf{J}_{P,\x}$ of partial
derivatives of the components of the map $P$
\begin{equation*}
  T_\x P(\h) = \mathbf{J}_{P,\x}\cdot \h.
\end{equation*}

We assume for the remainder of this section that the map $P:V\to U$ is
differentiable for all $\x\in V$. The derivative defines a map from $V$ to
linear bounded maps from $V$ to $U$. We further assume $U$ and $V$ are finite
dimensional. Then the space of linear maps from $V$ to $U$ is isomorphic to the
tensor product $U\otimes V^*$, where the isomorphism is given by the
tensor contraction, sending a simple tensor $\uu\otimes f\in U\otimes
V^*$ to a linear map
 \begin{equation}
   \label{eq:lin_tenzor}
   \uu\otimes f:\x \mapsto f(\x)\cdot \uu.
 \end{equation}
The derivative defines a map
\begin{eqnarray}
  \label{eq:odvod_preslikava}
  \D P&:& V\to U\otimes V^*\\
  \D P&:& \x \mapsto T_\x P.
\end{eqnarray}
One can consider the differentiability of the derivative itself $\D P$ by
looking at it as a map \eqref{eq:odvod_preslikava}. This leads to the definition
of the higher derivatives.
\begin{definition}[Higher derivatives]
  \label{def:higher_derivatives}
  Let $P:V\to U$ be a map from the vector space $V$ to the vector space $U$. 
The derivative $\D^k P$ of order $k$ of the map $P$ is the map
\begin{eqnarray}\label{eq:partial}
    \label{eq:visji_odvod}
    \D^kP&:&V\to U\otimes (V^*)^{\otimes k}\\
    \D^kP&:&\x\mapsto T_\x\left( \D^{k-1}P \right)
  \end{eqnarray}
\end{definition} 
\begin{remark}
  For the sake of clarity, we assumed in the definition above, that the map $P$
  as well as all its derivatives are differentiable at all points $\x$. If this
  is not the case, 
  definitions above can be done locally, which would introduce mostly technical
  difficulties. 
\end{remark}
Let $\e_1,\ldots,\e_n$ be a basis of $U$ and $x_1,\ldots x_m$ the basis of
$V^*$. Denote by $P_i=x_i\circ P$ the $i-th$ component of the map
$P$ according to the basis $\{\e_i\}$ of $U$.
Then $\D^kP$  can be defined in terms of
directional (partial) derivatives by the formula
\begin{equation}\label{eq:d}
  \partial^kP=\sum_{\forall_{i,\alpha}}\frac{\partial^k P_i}{\partial
      x_{\alpha_1}\ldots \partial x_{\alpha_k}}\e_i\otimes
    dx_{\alpha_1}\otimes\ldots \otimes dx_{\alpha_k}.
\end{equation}

\subsection{Differentiable Programs}

We want to be able to represent the derivative of a computer program in an
\emph{Euclidean machine} as a program in the same \emph{Euclidean machine}. We define three subspaces of the memory space $\VV$, that
describe how different parts of the memory influence the final result of the
program.   

Denote by $\e_1,\ldots \e_n$ a standard basis of the memory space $\VV$ and by
$x_1,\ldots x_n$ the dual basis of $\VV^*$. The functions $x_i$ are coordinate
functions on $\VV$ and correspond to individual locations(variables) in the
program memory.

\begin{definition}
  For each program $P$ in the programming space $\F<\VV^\VV$,
  we define the \emph{input} or \emph{parameter space} $I_P<\VV$ and the
  \emph{output space} $O_P<\VV$ to be the minimal vector sub-spaces spanned by
  the standard basis vectors, such that the map $P_e$, defined by the following
  commutative diagram 
\begin{equation} 
    \label{eq:induced_map}
\begin{tikzcd}
  \VV \arrow{r}{P} & 
  \VV \arrow{d}{\mathrm{pr}_{O_P}}\\
  I_P \arrow[hook]{u}{\vec{i}\mapsto \vec{i}+\vec{f}} 
  \arrow{r}{P_e}& O_P
\end{tikzcd}
  \end{equation}
does not depend of the choice of the element 
$\vec{f}\in F_P=(I_P+O_P)^\perp$.
The space $F_P=(I_P+O_P)^\perp$ is called \emph{free space} of the program $P$.
\end{definition}

The variables $x_i$ corresponding to the standard
basis vectors spanning the parameter, output and free space are called
\emph{paramters} or \emph{input variables}, \emph{output variables} and
\emph{free variables} correspondingly. Free variables are those that are
left intact by the program and have no influence on the final result other than
their value itself. The output of the program depends only on the values
of the input variables and consists of variables that have changed during
the program. Input parameters and output values might overlap. 

The map $P_e$ is called the \emph{effective map} of the program $P$ and
describes the actual effect of the program $P$ on the memory,
ignoring the free memory. 

The derivative of the effective map is of interest, when we speak about
differentiability of computer programs. 
\begin{definition}[Automatically differentiable programs]
  \label{def:program_derivative}
  A program $P:\VV\to \VV$ is \emph{automatically differentiable} if there exist
  an embedding of the space $O_P\otimes I_P^*$ into the free space $F_P$, and a program $(1+\D P):\VV\to \VV$,
  such that its effective map is the map
  \begin{equation}
    \label{eq:program_derivative}
    P_e\oplus \D P_e:I_P\rightarrow O_P\oplus (O_P\otimes I^*).
  \end{equation}
  A program $P:\VV\to \VV$ is \emph{automatically differentiable of order $k$}
  if there exist a program $\sumd_k P:\VV\to \VV$,
  such that its effective map is the map
  \begin{equation}
    \label{eq:program_derivative_higher}
    P_e\oplus \D P_e\oplus \ldots \D^k P_e:I_P\rightarrow O_P\oplus \left(O_P\otimes I^*\right)\oplus\ldots \left( O_P\otimes \left( I_p^*\right)^{k\otimes} \right).
  \end{equation}
\end{definition}

If a program $P:\VV\to \VV$ is automatically differentiable then it is also
differentiable as a map $\VV\to\VV$. However only the derivative of program's
effective map can be implemented as a program, since the memory space is limited to $\VV$. 
To be able to differentiate a program to the $k$-th order, we have to calculate
and save all the derivatives of the orders $k$ and less.

%% file: latex/Differentiable_Programming_Spaces.tex
\section{Differentiable Programming Spaces}

The memory space of a program is rarely treated as more than a storage. But to endow the \emph{Euclidean machine} with added structure, this is precisely what to focus on. Loosely speaking, functional programming is described by monoids, and as such a tensor algebraic description of the memory space is the appropriate step to take in attaining the wanted structure.

\subsection{Memory space}

Motivated by the Definition
\ref{def:program_derivative}, we define
the \emph{memory space} for differentiable programs  as a sequence of vector spaces with
the recursive formula
\begin{eqnarray}
  \VV_0 &=& \VV\\
  \label{eq:universal_space}
  \VV_k &=& \VV_{k-1}+\left(\VV_{k-1}\otimes \VV^*\right).
\end{eqnarray}
Note that the sum is not direct, since some of the subspaces of $\VV_{k-1}$ and
$\VV_{k-1}\otimes \VV^*$ are naturally isomorphic and will be
identified\footnote{The spaces $\VV\otimes(\VV^*)^{\otimes (j+1)}$ and
  $\VV\otimes (\VV^*)^{\otimes j}\otimes \VV^*$ are naturally isomorphic and
  will be identified in the sum.
}.

  The space that satisfies the recursive formula (\ref{eq:universal_space}) is
  \begin{equation}
    \label{eq:k-th-virtual-space}
    \VV_k = \VV\otimes \left(K\oplus \VV^* \oplus (\VV^*\otimes \VV^*)\oplus\ldots
      (\VV^*)^{\otimes k}\right) = \VV\otimes T_k(\VV^*),
  \end{equation}
  where $T_k(\VV^*)$ is a subspace of \emph{tensor algebra} $T(\VV^*)$, consisting of
  linear combinations of tensors of rank less or equal $k$. This construction
  enables us to define all the derivatives as maps with 
  the same domain and codomain $\VV\to \VV\otimes T(\VV^*)$.

  As such, an arbitrary  element of the memory space $\bfW\in\VV_n$ is a mapping
\begin{equation}
\bfW:\VV\to\VV,
\end{equation}
defined as 
\begin{equation}
\bfW(\vv)=\bfw_0+\bfw_1\cdot\vv+\cdots+\bfw_n\cdot(\vv)^{\otimes n}\label{eq:Contraction},
\end{equation}
the sum of multiple contractions (where $\bfw_i\in\VV_i$). The expression \eqref{eq:Contraction} will be rigorously defined in Section \ref{sec:Vrsta}. With such a construction, the expansions and contractions of the memory space (reminiscent to the breathing of the stack) would hold meaning parallel to storing values; which is what motives the next definition.

\begin{definition}[Virtual memory space]\label{def:VV}
Let $(\VV,\F)$ be an Euclidean machine and let  

\begin{equation}
\VV_\infty = \VV\otimes T(\VV^*) = \VV\oplus
(\VV\otimes\VV^*)\oplus\ldots,\label{eq:virtual-memory}
\end{equation}
where $T(\VV^*)$ is the tensor algebra of the dual space $\VV^*$.
We call $\VV_\infty$ the \emph{virtual memory space} of a \emph{Euclidean machine} $(\VV,\F)$.
\end{definition}
The term virtual memory is used as it is only possible to embed certain subspaces of $\VV_\infty$ into memory space $\VV$, making it similar to
virtual memory as a memory management technique. 

We can extend each program $P:\VV\to \VV$ to the map on
universal memory space $\VV_\infty$ by setting the first component in the direct sum
\eqref{eq:virtual-memory} to $P$, and all other components to zero. Similarly
derivatives $\D^k P$ can be also seen as maps  from $\VV$ to $\VV_\infty$ by
setting $k$-th component in the direct sum \eqref{eq:virtual-memory} to $\D^k P$
and all others to zero.

\subsection{Differentiable Programming Spaces}

Let us define the following function spaces:
 \begin{equation}\label{eq:F_n}
  \F_n=\{f:\VV\to \VV\otimes T_n(\VV^*)\}
 \end{equation}
All of these function spaces can be seen as subspaces of $\F_\infty=\{f:\VV\to \VV\otimes
T(\VV^*)\}$, since $\VV$ is naturally embedded into $ \VV\otimes T(\VV^*)$. The
Fréchet derivative defines an operator on the space of smooth maps in $\F_\infty$\footnote{The operator $\D$ may be defined partially for other maps as
   well.}. We denote this operator $\D$. The image of any map
 $P:\VV\to \VV$ by operator $\D$ is its first derivative, while the higher order
 derivatives are just powers of operator $\D$ applied to $P$.
 Thus $\D^k$ is a mapping between function spaces $\eqref{eq:F_n}$
 \begin{equation}\label{eq:toFn+k}
 \D^k:\F^n\to\F^{n+k}.
 \end{equation}

 \begin{definition}[Differentiable programming space]\label{def:dP}
  A \emph{differentiable programming space} $\dP_0$ is any subspace of $\F_0$ such that
  \begin{equation}\label{eq:P}
  \D\dP_0\subset\dP_0\otimes T(\VV^*)
\end{equation}
The space $\dP_n<\F_n$ spanned by $\{\D^k\dP_0;\quad 0\le k\le n\}$ over $K$, is
called a differentiable programming space of order $n$. When all elements of
$\dP_0$ are analytic, we call $\dP_0$ an \emph{analytic programming space}. 
 \end{definition}
The definition of higher order differentiable programming spaces is justified by the following theorem. 
\begin{theorem}[Infinite differentiability]\label{izr:P}
  Any differentiable programming space $\dP_0$ is an
  infinitely differentiable programming space, meaning that
  \begin{equation}\label{eq:P_n}
      \D^k\dP_0\subset\dP_0\otimes T(\VV^*)
    \end{equation}
for any $k\in\mathbb{N}$.
\end{theorem}
\begin{proof} By induction on order $k$. For $k=1$ the claim holds by definition. Assume  $\forall_{P\in\dP_0}$,
  $\D^n\dP_0\subset\dP_0\otimes T(\VV^*)$. Denote by $P_{\alpha,k}^i$ the
  component of the $k$-th derivative for a multiindex $\alpha$  denoting the
  component of $T(\VV^*)$ and an index $i$ denoting the component of $\VV$.
  \begin{equation}\label{eq:inductionStep}
\D^{n+1}P_{\alpha,k}^i=\D(\D^n P^i_\alpha)_k\land(\D^n P^i_\alpha)\in\dP_0\implies \D(\D^n P^i_\alpha)_k\in \dP_0\otimes T(\VV^*)
  \end{equation}
  $$\implies$$
  $$\D^{n+1}\dP_0\subset\dP_0\otimes T(\VV^*)$$
Thus by induction, the claim holds for all $k\in \mathbb{N}$. 
\end{proof}

 \begin{corollary}\label{izr:P_n}
  A differentiable programming space of order $n$, $\dP_n:\VV\to \VV\otimes
  T(\VV^*)$, can be embedded into the tensor 
  product of the function space $\dP_0$ and the space $T_n(\VV^*)$ of
  multi-tensors of order less than equal $n$:
  \begin{equation}
    \label{eq:D_p_embed}
    \dP_n<\dP_0\otimes T_n(\VV^*).
  \end{equation}
 \end{corollary}
 
By taking the limit as $n\to \infty$, we consider 
  
  \begin{equation}
  \label{eq:P_algebra}
        \dP_\infty < \dP_0\otimes \T(\VV^*),
  \end{equation}
where $\T(\VV^*)=\prod_{k=0}^\infty (\VV^*)^{\otimes k}$ is the \emph{tensor series
  algebra}, the algebra of the infinite formal tensor series.\footnote{The
  tensor series algebra is a completion of the tensor algebra $T(\VV^*)$ in suitable topology.} We will call \eqref{eq:P_algebra} the tensor series algebra of the programming space.

\subsection{Virtual Tensor Machine}

We propose an abstract computational model that is capable of constructing differentiable programming spaces and provides a
framework for algebraic study of analytic properties of differentiable programs. 

Following from Theorem \ref{izr:P}, the tuple $(\VV, \dP_0)$ -- together with the structure of the tensor algebra $T(\VV^*)$ -- is sufficient for constructing differentiable programming spaces $\dP_\infty$, using linear combinations of elements of the tensor series algebra of the programming space $\dP_0\otimes T(\VV^*)$. This motivates the following definition.

\begin{definition}[Virtual tensor machine]\label{def:analyticMachine}
The tuple $M=\langle \VV,\dP_0\rangle$ is an analytic, infinitely  differentiable virtual machine, where
   
    \begin{itemize}
    \item
    $\VV$ is a finite dimensional vector space
    \item
    $\VV\otimes \T(\VV^*)$ is the virtual memory space
    \item
    $\dP_0$ is an analytic programming space over $\VV$.
    \end{itemize}
  \end{definition}

\noindent When composing contractions \eqref{eq:Contraction} of the memory with activation functions $\phi\in\dP$, we note that fully connected \emph{tensor networks},
\begin{equation} \label{eq:tenWord}
\NN(v)=\phi_k\circ W_k\circ\cdots\circ\phi_0\circ W_0(v),
\end{equation}
are basic programs in a virtual tensor machine (the vanilla fully connected neural network is  captured by the restriction $\forall_i(W_i\in\VV_1)$). The formulation \eqref{eq:tenWord} is trivially generalized to convolutional models, but is omitted here for brevity.

%% file: latex/Operational_Calculus_on_Programming_Spaces.tex
\section{Operational Calculus on Programming Spaces}\label{sec:operational}

By Corollary $\ref{izr:P_n}$ we may represent calculation of derivatives of the
map $P:\VV\to \VV$, with only one mapping $\sumd$. We define the operator
$\sumd_n$ as a direct sum of operators
 
 \begin{equation}\label{eq:DD}
    \sumd_n = 1+\D +\D^2 +\ldots + \D^n 
  \end{equation}
  
The image $\sumd_kP(\x)$ is a multi-tensor of order $k$, which is a direct sum of the map's value and all derivatives of order $n\le k$, all evaluated at the point $\x$:
\begin{equation}
  \label{eq:multi_odvod}
  \sumd_kP(\x) = P(\mathbf{x})+\partial_\mathbf{x} P(\mathbf{x}) + \partial^2_\mathbf{x} P(\mathbf{x}) + \ldots + \partial^k_\mathbf{x} P(\mathbf{x}).
\end{equation}
The operator $\sumd_n$ satisfies the recursive relation:
  \begin{equation}
    \label{eq:potenca(1+d)}
    \sumd_{k+1}=1+\D\sumd_{k},
  \end{equation}
that can be used to recursively construct programming spaces of arbitrary order. 
\begin{proposition}
Only explicit knowledge of $\sumd_1:\dP_0\to\dP_1$ is required for the
construction of $\dP_n$ from $\dP_1$. 
\end{proposition}
\begin{proof}
  The construction is achieved following the argument $\eqref{eq:inductionStep}$ of the proof of Theorem \ref{izr:P}, allowing simple implementation, as dictated by $\eqref{eq:potenca(1+d)}$. 
\end{proof}
        
\begin{remark}\label{rem:vTen}
Maps $\VV\otimes T(\VV^*)\to \VV\otimes T(\VV^*)$ are constructible using
tensor algebra operations and compositions of programs in $\dP_n$.
\end{remark}
       
\begin{definition}[Algebra product]
 For any bilinear map
 $$\cdot :\VV\times \VV\to \VV$$
 we can define a
 bilinear product $\cdot$ on $\VV\otimes \T(\VV^*)$ by the following rule on the
 simple tensors:
 \begin{equation}
   \label{eq:algebra_product}
   (\vv\otimes f_1\otimes\ldots f_k) \cdot (\uu\otimes g_1\otimes\ldots g_l)=
(\vv\cdot \uu)\otimes f_1\otimes\ldots f_k\otimes g_1\otimes\ldots g_l 
 \end{equation}
extending linearly on the whole space $\VV\otimes\T(\VV^*)$
\end{definition}
\begin{theorem}[Programming algebra]\label{izr:alg}
 For any bilinear map $\cdot :\VV\times \VV\to \VV$ an infinitely-differentiable
 programming space $\dP_\infty$ is a function algebra, with the product defined
 by \eqref{eq:algebra_product}.
\end{theorem}

 \subsection{Tensor Series Expansion}\label{sec:Vrsta}

With the fundamentals of our framework established, we can begin to implement operators within its algebra. We begin by implementing an operator that shifts the program from its initial value and can later be used for the implementation of iterators and composers.

In the space spanned by the set $\DD^n=\{\D^k;\quad 0\le k\le n\}$ over a field $K$, such an operator can be defined as
 \begin{equation*}
  e^{h\D}=\sum\limits_{n=0}^{\infty}\frac{(h\D)^n}{n!}
 \end{equation*}
 In coordinates, the operator $e^{h\D}$ can be written as a
 series over all multi-indices $\alpha$
 \begin{equation}\label{eq:e^d}
  e^{h\D}=\sum\limits_{n=0}^{\infty}\frac{h^n}{n!}\sum_{\forall_{i,\alpha}}\frac{\partial^n}{\partial
        x_{\alpha_1}\ldots \partial x_{\alpha_n}}\e_i\otimes
      dx_{\alpha_1}\otimes\ldots \otimes dx_{\alpha_n}.
 \end{equation}
The operator $e^{h\D}$ is a mapping between programming spaces $\eqref{eq:F_n}$
 \begin{equation*}
  e^{h\D}:\dP\to\dP_\infty,
 \end{equation*}
in which partial applications can be made complete
  \begin{equation}\label{eq:specProg}
    e^{h\D}:\dP\to \Big\{\VV\to \VV\otimes \T(\VV^*)\Big\},
  \end{equation}
by taking the image of the map $e^{h\D}(P)$ at a certain point $\vv\in \VV$.  
Thus, we construct a map from the space of programs,
to the space of polynomials using  $\eqref{eq:specProg}$. Note that the space of
multivariate polynomials 
$\VV\to K$ is isomorphic to symmetric algebra $S(\VV^*)$, which is in turn a
quotient of tensor algebra $T(\VV^*)$.
To any element of
 $\VV\otimes T(\VV^*)$ one can attach corresponding element of $\VV\otimes S(\VV^*)$
 namely a polynomial map  $\VV\to \VV$. Thus, similarly to \eqref{eq:P_algebra}, we consider the completion of the symmetric algebra $S(\VV^*)$ as the \emph{formal power series} $\Ss(\VV^{*})$, which is in turn isomorphic to a quotient of \emph{tensor series algebra} $\T(\VV^*)$. This leads to 
 \begin{equation}\label{eq:pToPol}
  e^{h\D}:\dP\to \Big\{\VV\to \VV\otimes \Ss(\VV^*)\Big\}.
 \end{equation}
 For any element $\vv_0\in \VV$, the expression $e^{h\D}(\cdot,\vv_0)$ is a map $\dP\to
 \VV\otimes \Ss(\VV^*)$, mapping a program to a formal power series (by switching the order of partial applications in \eqref{eq:pToPol}).

 We can express the
 correspondence between multi-tensors in $\VV\otimes T(\VV^*)$ and polynomial maps
 $\VV\to \VV$ given by multiple contractions for all possible indices. For a simple tensor $\uu\otimes
 f_1\otimes\ldots\otimes f_n\in \VV\otimes(\VV^*)^{\otimes n}$ the contraction by
 $\vv\in \VV$ is given by applying co-vector $f_n$ to $\vv$ \footnote{For order
   two tensors from $\VV\otimes\VV^*$ the contraction corresponds to
matrix vector multiplication.}
 \begin{equation}
   \label{eq:contraction}
 \uu\otimes f_1\otimes\ldots\otimes f_n\cdot \vv = f_n(\vv) \uu\otimes f_1\otimes\ldots f_{n-1}.
\end{equation}
By taking contraction multiple times, we can attach a monomial map to a
simple tensor by  
 \begin{equation}
   \label{eq:monomial}
 \uu\otimes f_1\otimes\ldots\otimes f_n\cdot (\vv)^{\otimes n} = f_n(\vv)f_{n-1}(\vv)\cdots f_1(\vv) \uu,
\end{equation}
Both contractions \eqref{eq:contraction} and \eqref{eq:monomial} are extended
by linearity to spaces $\VV\otimes \left(\VV^*\right)^{\otimes n}$ and further
to $\VV\otimes T(\VV^*)$.\footnote{Note that the simple order one tensor
  $\uu\in\VV$ can not be contracted by the vector $\vv$. To be consistent we
  define $\uu\cdot \vv = \uu$ and attach a constant map
  $\vv\mapsto \uu$ to order zero tensor $\uu$. The extension of
  \eqref{eq:monomial}
  to $\VV\otimes T(\VV^*)$ can be seen as a generalization of the affine map,
  where the zero order tensors account for translation.}
For a multi-tensor $\bfW=\bfw_0+\bfw_1+\ldots+\bfw_n\in\VV\otimes T_n(\VV^*)$,
where $\bfw_k\in\VV\otimes\left( \VV^*\right)^{\otimes k}$, applying the
contraction by a vector $\vv\in \VV$ multiple times yields a polynomial map
\begin{equation}
  \label{eq:polynomial_tensor}
  \bfW(\vv) = \bfw_0+\bfw_1\cdot \vv+\ldots+\bfw_n\cdot (\vv)^{\otimes n}.
\end{equation}
\begin{theorem}\label{izr:e^d}
  For a program $P\in\dP$  the expansion into an infinite tensor series
  at the point $\vv_0\in \VV$ is expressed by multiple contractions 
  \begin{multline}\label{eq:tenzorVrsta}
  P(\vv_0+h\vv) = \Big((e^{h\D}P)(\vv_0)\Big)(\vv)
  = \sum_{n=0}^\infty\frac{h^n}{n!}\D^nP(\vv_0)\cdot (\vv^{\otimes n})\\
  = \sum_{n=0}^\infty \frac{h^n}{n!}\sum_{\forall_{i,\alpha}}\frac{\partial^nP_i}{\partial
        x_{\alpha_1}\ldots \partial x_{\alpha_n}}\e_i\cdot
      dx_{\alpha_1}(\vv)\cdot\ldots \cdot dx_{\alpha_n}(\vv).
  \end{multline}
\end{theorem}
 
 \begin{proof}
We will show that $\frac{d^n}{dh^n}\text{(LHS)}|_{h=0}=\frac{d^n}{dh^n}\text{(RHS)}|_{h=0}$. Then $\text{LHS}$ and $\text{RHS}$ as functions
of $h$ have coinciding Taylor series and are therefore equal.\\
 $\implies$
 
 $$\left. \frac{d^n}{dh^n}P(\vv_0+h\vv)\right|_{h=0}=\D^n P(\vv_0)(\vv)$$
 $\impliedby$
 $$\left. \frac{d^n}{dh^n}\left((e^{h\D})(P)(\vv_0)\right)(\vv)\right|_{h=0}=
\left. \left((\D^n e^{h\D})(P)(\vv_0)\right)(\vv)\right|_{h=0}$$
 $$\land$$
 $$\left. \D^ne^{h\D}\right| _{h=0}=\left. \sum\limits_{i=0}^{\infty}\frac{h^i\D^{i+n}}{i!}\right|_{h=0}=\D^n$$
 $$\implies$$
 $$\left(\D^n(P)(\vv_0)\right)\cdot(\vv^{\otimes n})$$
 \end{proof}
 \begin{remark}\label{konvVrst}
 Theorem \ref{izr:e^d} can be generalized to convolutions using the Volterra series \cite{volterra}.
 \end{remark}
It follows trivially from the above theorem that the operator $e^\D$ is contained strictly within the language
\begin{equation}
  e^{h\D}(\dP_0)\subset\dP_0\otimes \T(\VV^*),
\end{equation}
 and that the operator $e^{h\D}$ is an automorphism of the programming algebra $\dP_\infty$, 
\begin{equation}\label{eq:prod}
  e^{h\D}(p_1\cdot p_2)=e^{h\D}(p_1)\cdot e^{h\D}(p_2)
 \end{equation}
 where $\cdot$ stands for any bilinear map. 
 \begin{remark}[Generalized shift operator]\label{rmrk:genShift}
 The operator $e^{h\D}:\dP\times \VV\to \VV\otimes \T(\VV^*)$ evaluated at $h=1$
 is a broad generalization of the shift operator \cite{OpCalc}.
 \end{remark}
 
 For a specific $\vv_0\in\VV$, the generalized shift operator is denoted by
 \begin{equation*}
 e^\D\vert_{\vv_0}:\dP\to \VV\otimes \T(\VV^*)
 \end{equation*}
 When the choice of $\vv_0\in\VV$ is arbitrary, we omit it from expressions for brevity.




\subsection{Operator of Program Composition}\label{sec:compsition}

In this section we implement the operator of program composition within the constructed algebraic language. Such a \emph{composer} can than be used to implement the analog of the \emph{U combinator} (which facilitates recursion) and other constructs. Furthermore, due to the differentiable nature of the language, such a \emph{composer} generalizes both forward (e.g. \cite{PcAD}) and reverse (e.g. \cite{ReverseAD}) mode of
 automatic differentiation of arbitrary order, unified under a single operator. Upon completion we will demonstrate how to perform calculations on the operator level, before they are applied to a particular programming space, which serves as a level of abstraction over the tensor series algebra of the memory space.
 
 \begin{theorem}[Program composition]\label{izr:kompo}
 Composition of maps $\dP$ is expressed as
 \begin{equation}\label{eq:kompo}
 e^{h\D}(f\circ g)=\exp(\D_fe^{h\D_g})(g,f)
 \end{equation}
 where $\exp(\D_fe^{h\D_g}):\dP\times\dP\to\dP_\infty$ is an operator on pairs
 of maps $(g,f)$, where $\D_g$ is differentiation operator applied to the first
 component $g$, and $\D_f$ to the second component $f$. 
 \end{theorem}
 
\begin{proof}
  We will show that $\frac{d^n}{dh^n}\text{(LHS)}|_{h=0}=\frac{d^n}{dh^n}\text{(RHS)}|_{h=0}$. Then $\text{LHS}$ and $\text{RHS}$ as functions
  of $h$ have coinciding Taylor series and are therefore equal.\\
 $\implies$
 $$\lim\limits_{\lVert h\rVert\to 0}(\frac{d}{dh})^ne^\D(f\circ g)=\lim\limits_{\lVert h\rVert\to 0}\D^ne^{h\D}(f\circ g)$$
 $$\implies$$
 \begin{equation}\label{eq:kompproof1}
 \D^n(f\circ g)
 \end{equation}
 
 $\impliedby$
 $$\exp(\D_fe^{h\D_g})=\exp\left(\D_f\sum\limits_{i=0}^{\infty}\frac{(h\D_g)^i}{i!}\right)=\prod_{i=1}^{\infty}e^{\D_f\frac{(h\D_g)^i}{i!}}\Big(e^{\D_f}\Big)$$
 $$\implies$$
 $$\exp(\D_fe^{h\D_g})(g,f)=\sum\limits_{\forall_n}h^n\sum\limits_{\lambda(n)}\prod\limits_{k\cdot l\in\lambda}\Big(\frac{\D_f\D_g^l(g)}{l!}\Big)^k\frac{1}{k!}\Big(\Big(e^{\D_f}\Big)f\Big)$$
 where $\lambda(n)$ stands for the partitions of $n$. Thus
 \begin{equation}\label{eq:dComposite}
 \lim\limits_{\lVert h\rVert\to 0}(\frac{d}{dh})^n\exp(\D_fe^{h\D_g})=\sum\limits_{\lambda(n)}n!\prod\limits_{k\cdot l\in\lambda}\Big(\frac{\D_f\D_g^l(g)}{l!}\Big)^k\frac{1}{k!}\Big(\Big(e^{\D_f}\Big)f\Big)
 \end{equation}
 taking into consideration the fact that $e^{\D_f}(f)$ evaluated at a point $\vv\in \VV$ is the same as evaluating $f$ at $\vv$, the expression \eqref{eq:dComposite} equals \eqref{eq:kompproof1} by Faà di Bruno's formula.
   \begin{equation}\label{eq:dCompositePoint}
   \lim\limits_{\lVert h\rVert\to 0}(\frac{d}{dh})^n\exp(\D_fe^{h\D_g})=\sum\limits_{\lambda(n)}n!\prod\limits_{k\cdot l\in\lambda}\Big(\frac{\D_f\D_g^l(g(v))}{l!}\Big)^k\frac{1}{k!}\Big(f(g(\vv))\Big)
   \end{equation}
 \end{proof}       
 The Theorem $\ref{izr:kompo}$ enables an invariant implementation of the operator of program composition (i.e. the \emph{composer}) in $\dP_n$, expressed as a tensor series through $\eqref{eq:kompo}$ and $\eqref{eq:dComposite}$. 
 
  The operator of program composition 
  \begin{equation}\label{eq:opGenKompo}
 \exp(\D_fe^{h\D_g}): \dP\to\dP\to\dP_\infty,
  \end{equation}
  allows two kinds of partial applications. The operator resulting from fixing the second map $g$ in \eqref{eq:opGenKompo},
  \begin{equation}
   \exp(\D_fe^{h\D_g})(\cdot,g)=g^*\left( e^{h\D} \right)\label{eq:opKompo}
 \end{equation}
  is the pullback of the generalized shift
  operator $e^{h\D}$ through $g$. While the operator resulting from fixing the first map $f$ in \eqref{eq:opGenKompo},
  \begin{equation}
 \exp(\D_fe^{h\D_g})(f,\cdot)=f_*\left( e^{h\D} \right)\label{eq:opKompoForward}
 \end{equation} 
 is the push-forward of the generalized shift operator $e^{h\D}$ through $f$.
This also generalizes the \emph{U combinator} to its forward and backward modes, by restricting the \emph{composers} \eqref{eq:opGenKompo} domain to a single function (i.e. $f$ and $g$ are the same mapping).
 
  \begin{remark}[Unified AD]\label{trd:reverseForward}
  Because of \eqref{eq:induced_map} and \eqref{eq:program_derivative} every program can be seen as $P=P_n\circ\ldots P_1$. Thus applying the operators
  $\exp(\D_fe^{h\D_g})(\cdot,P_i)$ from $i=1$ to $i=n$ and projecting onto the space
  spanned by $\{1,\D\}$ is equivalent to forward mode automatic differentiation,
  while applying the 
  operators $\exp(\D_fe^{h\D_g})(P_{n-i+1},\cdot)$ in reverse order (and
  projecting) is equivalent to reverse mode automatic differentiation.
 \end{remark}

   \begin{corollary}\label{izr:komp_homo}
   The operator $e^{h\D}$ commutes with composition over $\dP$
   \begin{equation*}
   e^{h\D}(p_2\circ p_1)=e^{h\D}(p_2)\circ e^{h\D}(p_1)
   \end{equation*}
   \end{corollary}
   
   \begin{proof}
   Follows from $\eqref{eq:pToPol}$ and Theorem $\ref{izr:kompo}$.
   \end{proof}

Such calculations can be made easier, by completing them on
the level of operators, thus avoiding the need to manipulate tensor series. This serves as a level of abstraction over the tensor series algebra of the memory space.

The derivative $\frac{d}{dh}$ of $\eqref{eq:opKompo}$ is 
 \begin{equation}\label{eq:dexp}
 \frac{d}{dh}\exp(\D_fe^{h\D_g})(g)=\D_f(\D_gg)e^{h\D_g}\exp(\D_fe^{h\D_g})(g).
 \end{equation} 
 We note an important distinction to the operator $e^{h\D_g}$, the derivative of which is
 \begin{equation}\label{eq:de}
\frac{d}{dh}e^{h\D_g}=\D_ge^{h\D_g}.
 \end{equation}
 We may now compute derivatives (of arbitrary order) of the \emph{composer} itself.

\subsection{Example of an Operator Level Computation}\label{sec:example}

 For illustrative purposes we compute the second derivative of the \emph{composer} \eqref{eq:kompo}
 $$\left(\frac{d}{dh}\right)^2\exp\left(\D_fe^{h\D_g}\right)(g)=\frac{d}{dh}\left(\D_f(\D_gg)e^{h\D_g}\exp\left(\D_fe^{h\D_g}\right)(g)\right)$$
 which is by equations $\eqref{eq:dexp}$ and $\eqref{eq:de}$, using algebra and correct applications equal to
 \begin{equation}\label{eq:d^2comp}
 \left(\D_f(\D^2_gg)\right)e^{h\D_g}\exp(\D_fe^{h\D_g})(g)+(\D^2_f(\D_gg)^2)e^{2h\D_g}\exp(\D_fe^{h\D_g})(g)
 \end{equation}
 The operator is always shifted to the evaluating point $\eqref{eq:specProg}$ $\vv\in \VV$, thus, only the behaviour in the limit as $h\to 0$ is of importance. Taking this limit in the expression $\eqref{eq:d^2comp}$ we obtain the operator
 \begin{equation*}
  \left(\D_f(\D^2_gg)+\D^2_f(\D_gg)^2\right)\exp(\D_f):\dP\to\D^2\dP(g)
 \end{equation*}
 
 Thus, without imposing any additional rules, we computed the operator of the second derivative of composition with $g$, directly on the level of operators. The result of course matches the equation $\eqref{eq:dComposite}$ for $n=2$.
 
 As it is evident from the example, calculations using operators are far
 simpler, than direct manipulations of tensor series. This enables a simpler
 implementation that functions over arbitrary programming spaces. In
 the space that is spanned by $\{\D^n\dP_0\}$ over $K$, derivatives of
 compositions may be expressed solely through the operators, using only the
 product rule $\eqref{eq:prod}$ and the derivative of the general shift operator
 $\eqref{eq:de}$. Thus, explicit knowledge of rules for differentiating
 compositions is unnecessary, as it is contained in the structure of the
 operator $exp(\D_fe^{h\D_g})$ itself, which is differentiated using standard
 rules, as shown by this example. 

 Similarly higher derivatives of the \emph{composer} can be computed on the
 operator level
 \begin{equation}\label{eq:dkompo}
 \D^n(f\circ g)=\left.\left(\frac{d}{dh}\right)^n\exp\left(\D_fe^{h\D_g}\right)(g,f)\right|_{h=0}.
 \end{equation}

 \subsection{Automatically differentiable derivatives}\label{sec:orderReduction}
 
 The ability to use $k$-th derivative of a program $P_1\in\dP$ as part of a differentiable program $P_2\in\dP$ appears useful in many fields (e.g. \cite{StatMC}). For that to be sensible, we must be able to treat the ($k$-th) derivative itself as a differentiable program $P^{\prime k}\in\dP$. This is what motivates the following theorem. 

\begin{theorem}[Order reduction]\label{izr:reductionMap}
There exists a reduction of order map $\phi:\dP_n\to \dP_{n-1}$, such that the
following  diagram commutes
\begin{equation}\label{eq:reductionMap}
\begin{tikzcd}
  \dP_n \arrow{r}{\phi} \arrow{d}{\D} & 
  \dP_{n-1} \arrow{d}{\D}\\
  \dP_{n+1} \arrow{r}{\phi} & 
  \dP_{n}
\end{tikzcd}
\end{equation}
satisfying
\begin{equation*}
\forall_{P_1\in\dP_0}\exists_{P_2\in\dP_0}\Big(\phi^k\circ e^\D_n(P_1)=e^\D_{n-k}(P_2)\Big)
\end{equation*}
for each $n\ge 1$, where $e^\D_n$ is the projection of the operator $e^\D$ onto the set $\{\D^n\}$.
\end{theorem}  
\begin{corollary}[Differentiable derivative]\label{cor:extraxtDerivatives}
By Theorem \ref{izr:reductionMap}, $n$-differentiable $k$-th derivatives of a program $P\in\dP_0$ can be extracted by
\begin{equation*}
^{n}P^{k\prime}=\phi^k\circ e^\D_{n+k}(P)\in\dP_n
\end{equation*}
\end{corollary}

Thus, by corollary \ref{cor:extraxtDerivatives}, the writing of differentiable programs that act on derivatives of other programs is well defined within the language. This is a crucial feature, as stressed by other authors \cite{AD1,AD2}. Note that in order to use $k$-th derivative of $P_2$ in an $n$-differentiable program $P_1$, then $P_2$ must have been $(k+n)$-differentiable before $\phi^k$ was applied to it.

%% file: latex/Iterators.tex
\subsection{Iterators and Iterating Velocity}

  Iterator is an operator, composing a program $p:\VV\to \VV\in\dP$ with itself. For ease of expression, we denote the $n$-th iterate of a program $p\in\dP$, as $p^n$, as it is possible to view iteration as compositional exponentiation. In this view, one may seek to explore the relation between the value of the $n$-th iterate and $n$. With a form which expresses the iterate as a function of $n$, one could inquire its rate of change in relation to it and investigate fractional iterations, akin to fractional powers of other operators \cite{komatsu1966fractional}.
  
Let $\mathcal{I}_p$ be the monoid under composition $\circ$ 
  \begin{equation}\label{eq:iter_def}
  \mathcal{I}_p=\{p^n:\VV\to \VV;\quad p\left(\vv_f\right)=\vv_f\},
  \end{equation}
  generated by $p\in\dP:\VV\to \VV$ with some fixed point $\vv_f\in \VV$; note that any terminating program has a fixed point. We than turn towards analysing the structure of \eqref{eq:iter_def} in relation to $n$, the number of iterations.

Let $\mathcal{C}_p:h\to h\circ p$ be the operator of composition with $p$ and assume $h$ to be the solution of the eigen equation
\begin{equation}\label{eq:kh}
  \mathcal{C}_p(h)=\Lambda \cdot h\iff h\left(p\left(\vv\right)\right)=\Lambda\cdot h\left(\vv\right).
  \end{equation}
It is clear that the composition of $h$ with $p^n$ is such, 
  \begin{equation} 
  h\left(p^n\left(\vv\right)\right)=\Lambda^n\cdot h\left(\vv\right)
  \end{equation}
that in the image of $h$, iterations of $p$ become multiplication with the eigen matrix $\Lambda$. This allows us to generalize the notion of an iteration from the integers, $n\in\mathbb{N}$, to the reals, $x\in\mathbb{R}$, and consider fractional iterations by
\begin{equation}
  p^x\left(\vv\right)=h^{-1}\left(\Lambda^x\cdot h\left(\vv\right)\right),
\end{equation}
assuming $h^{-1}$ exists. Note that we can express the eigen matrix $\Lambda$ by differentiating $\eqref{eq:kh}$ at the fixed point $\vv_f$,
$$\D h\left(p\left(\vv_f\right)\right)\cdot\D p\left(\vv_f\right)=\Lambda\cdot\D h\left(\vv_f\right)\implies\D p\left(\vv_f\right)=\Lambda.$$

With foundations established, we can proceed to inquire about the rate of change of the values of a program $p$ in relation to $n$, the number of iterations. Lets define the \emph{rate of change operator} $\Psi$ 
  \begin{equation}
    \Psi:p^n\to\D_np^n,
  \end{equation}
that maps an iterate $p^n$ to its \emph{iterating velocities} $\D_n p^n$. Of course
  \begin{equation}
  \Psi\left(p^n\right)\left(\vv_f\right)=0,
  \end{equation}
the iterating velocities of any iterate at the fixed point $\vv_f$ are constantly zero, which is deduced from the $\eqref{eq:kh}$ and reassures our intuition. Next, we introduce a change of variables $\Lambda= e^\nu$ for mathematical convenience and proceed towards computing the iterating velocity.
  $$\D_nh\left(p^n\right)=\D_n\left(e^{\nu n}\cdot h\right)$$
  $$\implies$$
  $$\D h\left(p^n\right)\cdot\D_np^n=\nu\cdot e^{\nu n}\cdot h \land e^{\nu n}\cdot h=h\left(p^n\right)$$
  $$\implies$$
  \begin{equation}\label{eq:iter_vel}
  \Psi=\nu\cdot(\D h)^{-1}\cdot h
  \end{equation}
The rate of change operator $\Psi$ and iterating velocities
\footnote{Higher derivatives can be derived by induction.}
\eqref{eq:iter_vel} can be used to study iterated processes, which feature prominently in machine learning; ex. examine the importance of continued iteration and aid decisions on early stopping.

The computation of the eigen map $h$ \eqref{eq:kh} was solved by Bridges \cite{bridges2016solution} for any $p$ with a power series representation. This result is extended to tensor series by the isomorphism to their quotient. Hence, as we can expand any $p\in\dP$ into a tensor series by the use of the operator $e^\D$, by Theorem $\ref{izr:e^d}$ the result also holds for any $p\in\dP$.

\subsection{ReduceSum in the Language of Operational Calculus}

As a demonstration of the algebraic power over analytic conclusions inherent to our model, we examine the functional \emph{ReduceSum}, and derive its explicit form as a function of $n$, the number of its iterations, or upper bound, with special interest in the rate of change of the functional in relation to $n$; i.e. iterating velocity and its higher order counter parts (acceleration etc.). 

Let $\mathcal{S}^n$ denote the operator, that performs a linear shift of a program $p$ in the direction $\vv$, from its initial point $\vv_0$. By Theorem $\ref{izr:e^d}$ we have
  \begin{equation}\label{eq:sn=ed}
    \left(e^{n\D}\vert_{\vv_0}p\right)(\vv)=p(\vv_0+n\vv)\implies \mathcal{S}^n=e^{n\D}\vert_{\vv_0},
  \end{equation}
and thus clearly $\mathcal{S}^n\cdot\mathcal{S}^m=\mathcal{S}^{n+m}$ and $(\mathcal{S}^n+\mathcal{S}^m)(p)=\mathcal{S}^n(p)+\mathcal{S}^m(p)$, which we use to define the $n$-th reduction as
$$\mathcal{R}_+^n=(1+\mathcal{S}+\mathcal{S}^2+\cdots+\mathcal{S}^n),$$
that results in
   $$\mathcal{R}^n_+(p)(\vv)=\sum\limits_{h=0}^{n}p(\vv_0+h\vv)$$
upon application.

With this we turn towards computing with operators alone to harness the algebraic power of our framework. We write
$$(1+\mathcal{S}+\mathcal{S}^2+\cdots+\mathcal{S}^n)=1+\mathcal{S}(1+\mathcal{S}+\mathcal{S}^2+\cdots+\mathcal{S}^{n-1})$$
   $$\implies$$
   $$1-\mathcal{S}^n=\left(1-S\right)\mathcal{R}^{n-1}_+$$
   $$\implies$$
  \begin{equation}
\mathcal{R}^{n-1}_+=\left(1-\mathcal{S}^n\right)\left(\frac{1}{1-\mathcal{S}}\right),
  \end{equation}
where $\frac{1}{1-\Ss}$ is to be understood in the sense of formal tensor series. We will denote $(1-\mathcal{S}^n)$ by $\Big[\cdot\Big]^n_{\vv_0}$, recognizing that it represents the action of shifting the program in the direction of $\vv$ by a fraction of $n$, from its initial position $\vv_0$, and subtracting the two; while $\vv$ is yet to be applied. Taking Theorem \ref{izr:e^d} into account we write
  \begin{equation}\label{eq:Rn+Basic}
    \mathcal{R}^{n-1}_+=\Bigg[\D^{-1}\left(\frac{\D}{1-e^\D}\right)\Bigg]^n_{\vv_0}.
  \end{equation}
Note that while $\D^{-1}$ is undetermined, its composition with $\Big[\cdot\Big]^n_{\vv_0}$ is well defined on $\VV\otimes\T(\VV^*)$.
Also note, that the parenthesised expression, $\frac{\D}{1-e^\D}$, and $\D^{-1}$ are algebraic encodings of higher-order programs, which are to be expanded into explicit form upon application inside the tensor series algebra of the programming space $\dP\otimes \T(\VV^*)$. Doing so, we recognize
\begin{equation}
    	\frac{h\D}{1-e^{h\D}}=\sum\limits_{i=0}^{\infty}B_i\frac{(h\D)^i}{i!}
    \end{equation}
$B_i$ to be the $n$-th Bernoulli number. Thus the higher order program
$$\mathcal{R}^{n-1}_+:\dP\to\Big\{\VV\to \VV\otimes \T(\VV^*)\Big\}$$
is expressed as 
    \begin{equation}\label{eq:Rn+}
      \mathcal{R}^{n-1}_+=\Bigg[B_0\D^{-1}+\sum\limits_{i=1}^{\infty}B_i\frac{\D^{i-1}}{i!}\Bigg]_{\vv_0}^{n}.
    \end{equation}
Upon applying it to a program at a particular point $\vv\in \VV$ this becomes
\begin{equation}\label{eq:Rn+P}
  \mathcal{R}^{n-1}_+p(\vv)=\Bigg[B_0\D^{-1}p(t)+\sum\limits_{i=1}^{\infty}B_i\frac{\D^{i-1}p(t)}{i!}\Bigg]_{\vv_0}^{n\vv} \Big(\vv\Big),
    \end{equation}
where the evaluation at $\vv\in \VV$ as by \eqref{eq:polynomial_tensor}, performs the needed translation
\footnote{Both the translation and the shift are to be performed (at) by the same point $\vv\in\VV$.},
as the image of the operator $\mathcal{R}^{n-1}_+$ is an element of the tensor series algebra of the memory space $\VV\otimes \T(\VV^*)$.
Note that $n$, the number of iterations, is the only remaining free variable, as desired.
\begin{remark}
  When $p\in\dP$ is an univariate mapping, the expression \eqref{eq:Rn+P} recovers the Euler-Maclaurin integral formula \cite{apostol1999elementary}, which is demonstrated by applying the operator \eqref{eq:Rn+} to the function $x^m$
$$\mathcal{R}^n_+(x^m)=\frac{1}{m+1}\sum_{i=0}^m{m+1\choose i}B_i\cdot n^{m+1-i},$$
and producing the closed form solution.
\end{remark}
Furthermore, due to the operational algebra of higher order programs established by our model, we can compute the operator of iterating velocity (and higher order change) of the $n$-th iterate by differentiating the operator $\mathcal{R}^{n-1}_+$ itself. Reverting to the form \eqref{eq:Rn+Basic} and substituting \eqref{eq:sn=ed}, we have
\begin{equation}
  \frac{d^k}{dn^k}\mathcal{R}^{n-1}_+=\frac{d^k}{dn^k}\left(\left(1-e^{n\D}\right)\left(\frac{1}{1-e^\D}\right)\right)=\mathcal{S}^n\left(\frac{\D^n}{1-e^\D}\right),
   \end{equation}
where commutativity of shifting and differentiating was used. Noting that $\mathcal{S}^n$ simply shifts the operand in the direction of $\vv$ by a factor of $n$, the explicit form 
\begin{equation}
        \frac{d^k}{dn^k}\mathcal{R}^{n-1}_+\dP\vert_{n=N}(\vv)=\left(\sum\limits_{i=0}^{\infty}B_i\frac{\D^{N-1+i}\dP\left(\vv_0+N\vv\right)}{i!}\right)\Big(\vv\Big)
        \end{equation}
where the evaluation at $\vv\in \VV$ once again performs the needed translation.

%% file: latex/Conclusions.tex
\section{Conclusions}

In this paper we presented a theoretical model for differentiable programming.
Throughout the course of the paper we have shown the model to be a complete description of differentiable programming.
Furthermore, the innate algebraic structure of the framework supplements the descriptive power of a language with the ability to reason about the programs it implements, by way of operational calculus.
We believe operational calculus has a place in the evolution of computer science, where languages are to be endowed with algebraic constructs that hold power over analytic properties of the programs they implement.
These results hope to inspire other practitioners of differentiable programming to reach for operational calculus on their quest to further the field.